\documentclass[12pt,a4paper]{article}
\usepackage{fullpage}
\usepackage{amssymb,amsmath,stmaryrd}
\usepackage[mathscr]{eucal}
\usepackage{epsfig}
\usepackage{pstricks,pst-node,pst-text,pst-3d}

\newtheorem{definition}{Definition}
\newtheorem{theorem}{Theorem}
\newtheorem{proposition}{Proposition}

\newtheorem{lemma}{Lemma}

\def \bel {\mathrm{bel}}
\def \pl {\mathrm{pl}}
\def \N {\mathrm{N}}

\newcounter{remark}
\newenvironment{remark}
{\begin{quote}\textsc{Remark} \stepcounter{remark} \arabic{remark}:}
{\end{quote}}
\newcounter{example}
\newenvironment{example}
{\begin{quote}\textsc{Example} \stepcounter{example} \arabic{example}:}
{\end{quote}}
\newenvironment{proof}{\medskip\noindent \bf Proof: \rm}{\hspace*{\fill}
$\blacksquare$ \newline \medskip}  

\begin{document}

\title{Belief functions on lattices}

\author{Michel GRABISCH\\
Universit\'e Paris I -- Panth\'eon-Sorbonne\\
\normalsize email \texttt{Michel.Grabisch@lip6.fr}
}

\date{}

\maketitle

\begin{abstract}
We extend the notion of belief function to the case where the underlying
structure is no more the Boolean lattice of subsets of some universal set, but
any lattice, which we will endow with a minimal set of properties according to
our needs. We show that all classical constructions and definitions (e.g., mass
allocation, commonality function, plausibility functions, necessity measures
with nested focal elements, possibility distributions, Dempster rule of
combination, decomposition w.r.t. simple support functions, etc.) remain valid
in this general setting. Moreover, our proof of decomposition of belief
functions into simple support functions is much simpler and general than the
original one by Shafer.
\end{abstract}


\section{Introduction}
The theory of evidence, as established by Shafer \cite{sha76} after the work of
Dempster \cite{dem67}, and brought into a practically usable form by the works
of Smets in particular \cite{sme90,sme94}, has become a popular tool in
artificial intelligence for the representation of knowledge and making decision.
In particular, many applications in classification have been done
\cite{den95,den00a}. The main advantage over more traditional models based on
probability is that the model of Shafer allows for a proper representation of
ignorance.

On a mathematical point of view, belief functions, which are at the core of the
theory of evidence, possess remarkable properties, in particular their links
with the M\"obius transform \cite{rot64} and the co-M\"obius transform
\cite{gra98c,grla00a}, called \emph{commonality} by Shafer. Remarking that
belief functions are non negative isotone functions defined on the Boolean
lattice of subsets, one may ask if all these properties remain valid when more
general lattices are considered. The aim of this paper is precisely to
investigate this question, and we will show that amazingly they all remain
valid. A first investigation of this question was done by Barth\'elemy
\cite{bar00}, and our work will complete his results. We are not aware of other
similar works, except the one of Kramosil \cite{kra01}, where belief functions
are defined on Boolean lattices but take value in a partially ordered set, and
the notion of bi-belief proposed by Grabisch and Labreuche \cite{grla03}, where
the underlying lattice is $3^n$.

On an application point of view, one may ask about the usefulness of such a
generalization, apart from its mathematical beauty. A general answer to this is
that the objects we manipulate (events, logical propositions, etc.) may not form
a Boolean lattice, i.e., distributive and complemented. Thus, a study on a
weaker yet rich structure has its interest. Let us give some examples.
\begin{itemize}
\item Case where the universal set $\Omega$ is the set of possible outcomes,
  states of nature, etc. In the classical case, all subsets of $\Omega$ (called
  events) are considered, but it may happen that some events are not observable
  or realizable, meaningful, etc. Then, the structure of the events is no more
  the Boolean lattice $2^\Omega$. 
\item Case where the universal set $\Omega$ is the set of propositional
  variables, either true or false. As argued by Barth\'elemy \cite{bar00}, in
  non-classical logics, the set of propositions need not be $2^\Omega$, and as
  we will see later, probability theory applies as far as the lattice induced by
  propositional calculus is distributive, and this covers intuitionistic logic
  and paraconsistent logic. If distributivity does not hold, then belief
  functions appear as a natural candidate, since as it will be shown, belief
  functions can live on any lattice.
\item Case where the universal set is the set of players/agents in some
  cooperative game or multiagent situation. Subsets of $\Omega$ are called
  coalitions, and most of the time, it happens that some coalitions are
  infeasible, .i.e., they cannot form, due to some inherent impossibility
  depending on the context. For example, in voting situations, clearly all
  coalitions of political parties cannot form. The same holds for agents or
  players in general where some incompatibilities exist between them.   
\item Knowledge extraction and modeling: objects under study are often
  structured as lattices. For example, the popular Formal Concept Analysis of
  Ganter and Wille \cite{gawi99} build lattices of  concepts, from a matrix of
  objects described by qualitative attributes. 
\item Finally, in some cases, objects of interest are not subsets of some
  universal set.
  This is the case for example when one is interested into the collection of
  partitions of some set (again, this happens in game theory under the name
  ``game in partition function form'' \cite{thlu63}, and also in knowledge
  extraction where the fundamental problem is to partition attributes), or when
  objects of interest are ``bi-coalitions'' like for bi-belief functions. A
  bi-coalition is a pair of subsets with empty intersection, and it may
  represent the set of criteria which are satisfied and the one which are not
  satisfied.  
\end{itemize}

The paper is organized as follows. Section \ref{sec:back} recalls necessary
material on lattices and classical belief functions. Section \ref{sec:bella}
gives the main results on belief defined over lattices, while the last one
examine the case of necessity measures. 

Throughout the paper, we will deal with finite lattices.

\section{Background}
\label{sec:back}
\subsection{Lattices}
\label{sec:latti}
We begin by recalling necessary material on lattices (a good introduction on
lattices can be found in \cite{dapr90} and \cite{mon03}), in a finite setting. A
\emph{poset} is a set $P$ endowed with a partial order $\leq$ (reflexive,
antisymmetric, transitive). A \emph{lattice} $L$ is a poset such that for any
$x,y\in L$ their least upper bound $x\vee y$ and greatest lower bound $x\wedge
y$ always exist. For finite lattices, the greatest element of $L$ (denoted
$\top$) and least element $\bot$ always exist. $x$ \emph{covers} $y$ (denoted
$x\succ y$) if $x> y$ and there is no $z$ such that $x>z>y$. Let $P$ be a poset,
$Q\subseteq P$ is a \emph{downset} if for any $y\in P$ such that $y\leq x$,
$x\in Q$, then $y\in Q$. The set of all downsets of $P$ is denoted by
$\mathcal{O}(P)$.

A \emph{linear lattice}, or \emph{chain}, is such that $\leq$ is a total order.
A chain $C$ in $L$ is \emph{maximal} if no element $x\in L\setminus C$ can be
added so that $C\cup\{x\}$ is still a chain.

Lattices can be represented by their \emph{Hasse diagram}, where nodes are
elements of the lattice, and there is an edge between $x$ and $y$, with $x$
above $y$, if and only if $x\succ y$. Fig. \ref{fig:lat1} shows three
lattices. The middle and right ones are two different diagrams of the lattice
of subsets of $\{1,2,3\}$ ordered by inclusion.
\begin{figure}[htb]
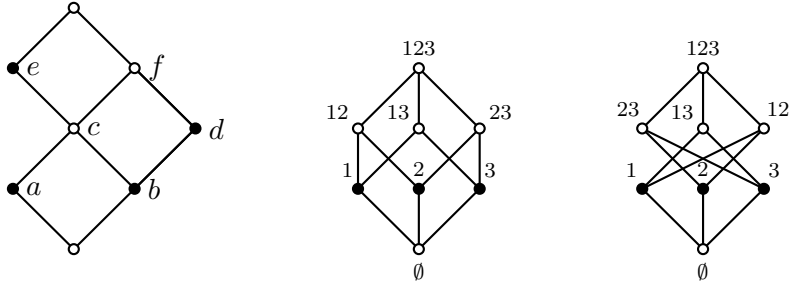

\begin{center}
\psset{unit=0.8cm}
\pspicture(0,0)(3,4)
\pspolygon(1,0)(3,2)(2,3)(0,1)
\pspolygon(2,1)(3,2)(1,4)(0,3)
\pscircle[fillstyle=solid](1,0){0.1}
\pscircle[fillstyle=solid,fillcolor=black](0,1){0.1}
\pscircle[fillstyle=solid](1,2){0.1}
\pscircle[fillstyle=solid,fillcolor=black](2,1){0.1}
\pscircle[fillstyle=solid,fillcolor=black](3,2){0.1}
\pscircle[fillstyle=solid,fillcolor=black](0,3){0.1}
\pscircle[fillstyle=solid](2,3){0.1}
\pscircle[fillstyle=solid](1,4){0.1}
\uput[0](1,2){\small $c$}
\uput[0](2,1){\small $b$}
\uput[0](0,1){\small $a$}
\uput[0](3,2){\small $d$}
\uput[0](0,3){\small $e$}
\uput[0](2,3){\small $f$}
\endpspicture
\hspace*{2cm}
\pspicture(0,0)(2,4)
\pspolygon(1,0)(0,1)(1,2)(2,1)
\pspolygon(1,1)(0,2)(1,3)(2,2)
\psline(0,1)(0,2)
\psline(1,0)(1,1)
\psline(2,1)(2,2)
\psline(1,2)(1,3)
\pscircle[fillstyle=solid](1,0){0.1}
\pscircle[fillstyle=solid,fillcolor=black](0,1){0.1}
\pscircle[fillstyle=solid,fillcolor=black](1,1){0.1}
\pscircle[fillstyle=solid,fillcolor=black](2,1){0.1}
\pscircle[fillstyle=solid](0,2){0.1}
\pscircle[fillstyle=solid](1,2){0.1}
\pscircle[fillstyle=solid](2,2){0.1}
\pscircle[fillstyle=solid](1,3){0.1}
\uput[-90](1,0){\scriptsize  $\emptyset$}
\uput[120](0,1){\scriptsize  $1$}
\uput[90](1,1){\scriptsize  $2$}
\uput[60](2,1){\scriptsize  $3$}
\uput[135](0,2){\scriptsize  $12$}
\uput[135](1,2){\scriptsize  $13$}
\uput[45](2,2){\scriptsize  $23$}
\uput[90](1,3){\scriptsize  $123$}
\endpspicture
\hspace*{2cm}
\pspicture(0,0)(2,4)
\psline(1,0)(0,1)
\psline(1,0)(1,1)
\psline(1,0)(2,1)
\psline(1,3)(0,2)
\psline(1,3)(1,2)
\psline(1,3)(2,2)
\psline(0,1)(1,2)
\psline(0,1)(2,2)
\psline(1,1)(0,2)
\psline(1,1)(2,2)
\psline(2,1)(0,2)
\psline(2,1)(1,2)
\pscircle[fillstyle=solid](1,0){0.1}
\pscircle[fillstyle=solid,fillcolor=black](0,1){0.1}
\pscircle[fillstyle=solid,fillcolor=black](1,1){0.1}
\pscircle[fillstyle=solid,fillcolor=black](2,1){0.1}
\pscircle[fillstyle=solid](0,2){0.1}
\pscircle[fillstyle=solid](1,2){0.1}
\pscircle[fillstyle=solid](2,2){0.1}
\pscircle[fillstyle=solid](1,3){0.1}
\uput[-90](1,0){\scriptsize  $\emptyset$}
\uput[120](0,1){\scriptsize  $1$}
\uput[90](1,1){\scriptsize  $2$}
\uput[60](2,1){\scriptsize  $3$}
\uput[60](2,2){\scriptsize  $12$}
\uput[135](1,2){\scriptsize  $13$}
\uput[120](0,2){\scriptsize  $23$}
\uput[90](1,3){\scriptsize  $123$}
\endpspicture
\end{center}
\caption{Examples of lattices}
\label{fig:lat1}
\end{figure}

Let $P,Q$ be two posets, and consider $f:P\rightarrow Q$. $f$ is \emph{isotone}
(resp. \emph{antitone}) if $x\leq y$ implies $f(x)\leq f(y)$ (resp. $f(x)\geq
f(y)$). $P$ and $Q$ are \emph{isomorphic} (resp. \emph{anti-isomorphic}),
denoted by $P\cong Q$ (resp. $P\cong Q^\partial$), if it exists a bijection $f$
from $P$ to $Q$ such that $x\leq y\Leftrightarrow f(x)\leq f(y)$ (resp.
$f(x)\geq f(y)$). Isomorphic posets have same Hasse diagrams, up to the
labelling of elements.

For any poset $(P,\leq$), one can
consider its \emph{dual} by inverting the order relation, which is denoted by
$(P,\leq^\partial)$ (or simply $P^\partial$ if the order relation is not
mentionned), i.e., $x\leq y$ if and only if $y\leq^\partial x$. \emph{Autodual}
posets are such that $P\cong P^\partial$ (i.e., they have the
same Hasse diagram). The lattices of Fig. \ref{fig:lat1} are all
autodual, and Fig. \ref{fig:lat2} shows their dual.
\begin{figure}[htb]
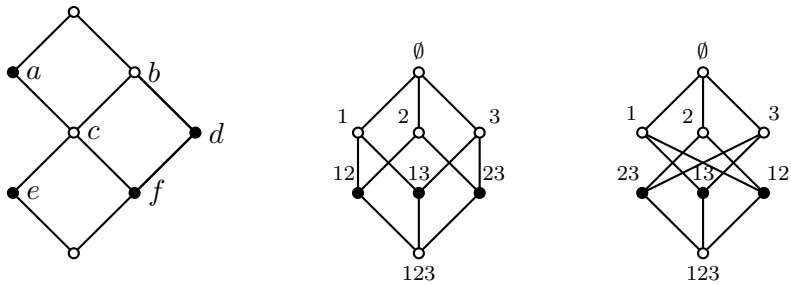

\begin{center}
\psset{unit=0.8cm}
\pspicture(0,0)(3,4)
\pspolygon(1,0)(3,2)(2,3)(0,1)
\pspolygon(2,1)(3,2)(1,4)(0,3)
\pscircle[fillstyle=solid](1,0){0.1}
\pscircle[fillstyle=solid,fillcolor=black](0,1){0.1}
\pscircle[fillstyle=solid](1,2){0.1}
\pscircle[fillstyle=solid,fillcolor=black](2,1){0.1}
\pscircle[fillstyle=solid,fillcolor=black](3,2){0.1}
\pscircle[fillstyle=solid,fillcolor=black](0,3){0.1}
\pscircle[fillstyle=solid](2,3){0.1}
\pscircle[fillstyle=solid](1,4){0.1}
\uput[0](1,2){\small $c$}
\uput[0](2,1){\small $f$}
\uput[0](0,1){\small $e$}
\uput[0](3,2){\small $d$}
\uput[0](0,3){\small $a$}
\uput[0](2,3){\small $b$}
\endpspicture
\hspace*{2cm}
\pspicture(0,0)(2,4)
\pspolygon(1,0)(0,1)(1,2)(2,1)
\pspolygon(1,1)(0,2)(1,3)(2,2)
\psline(0,1)(0,2)
\psline(1,0)(1,1)
\psline(2,1)(2,2)
\psline(1,2)(1,3)
\pscircle[fillstyle=solid](1,0){0.1}
\pscircle[fillstyle=solid,fillcolor=black](0,1){0.1}
\pscircle[fillstyle=solid,fillcolor=black](1,1){0.1}
\pscircle[fillstyle=solid,fillcolor=black](2,1){0.1}
\pscircle[fillstyle=solid](0,2){0.1}
\pscircle[fillstyle=solid](1,2){0.1}
\pscircle[fillstyle=solid](2,2){0.1}
\pscircle[fillstyle=solid](1,3){0.1}
\uput[-90](1,0){\scriptsize  $123$}
\uput[120](0,1){\scriptsize  $12$}
\uput[90](1,1){\scriptsize  $13$}
\uput[60](2,1){\scriptsize  $23$}
\uput[135](0,2){\scriptsize  $1$}
\uput[135](1,2){\scriptsize  $2$}
\uput[45](2,2){\scriptsize  $3$}
\uput[90](1,3){\scriptsize  $\emptyset$}
\endpspicture
\hspace*{2cm}
\pspicture(0,0)(2,4)
\psline(1,0)(0,1)
\psline(1,0)(1,1)
\psline(1,0)(2,1)
\psline(1,3)(0,2)
\psline(1,3)(1,2)
\psline(1,3)(2,2)
\psline(0,1)(1,2)
\psline(0,1)(2,2)
\psline(1,1)(0,2)
\psline(1,1)(2,2)
\psline(2,1)(0,2)
\psline(2,1)(1,2)
\pscircle[fillstyle=solid](1,0){0.1}
\pscircle[fillstyle=solid,fillcolor=black](0,1){0.1}
\pscircle[fillstyle=solid,fillcolor=black](1,1){0.1}
\pscircle[fillstyle=solid,fillcolor=black](2,1){0.1}
\pscircle[fillstyle=solid](0,2){0.1}
\pscircle[fillstyle=solid](1,2){0.1}
\pscircle[fillstyle=solid](2,2){0.1}
\pscircle[fillstyle=solid](1,3){0.1}
\uput[-90](1,0){\scriptsize  $123$}
\uput[120](0,1){\scriptsize  $23$}
\uput[90](1,1){\scriptsize  $13$}
\uput[60](2,1){\scriptsize  $12$}
\uput[60](2,2){\scriptsize  $3$}
\uput[135](1,2){\scriptsize  $2$}
\uput[120](0,2){\scriptsize  $1$}
\uput[90](1,3){\scriptsize  $\emptyset$}
\endpspicture
\end{center}
\caption{Dual of the lattices of Fig. \ref{fig:lat1}}
\label{fig:lat2}
\end{figure}

A lattice $L$ is \emph{lower semimodular} (resp. \emph{upper semimodular}) if
for all $x,y\in L$, $x\vee y\succ x$ and $x\vee y\succ y$ imply $x\succ x\wedge
y$ and $y\succ x\wedge y$ (resp. $x\succ x\wedge y$ and $y\succ x\wedge y$
imply $x\vee y\succ x$ and $x\vee y\succ y$). A lattice being upper and lower
semimodular is called \emph{modular}. The lattice is \emph{distributive} if
$(x\vee y)\wedge z= (x\wedge z)\vee(y\wedge z)$ holds for all $x,y,z\in L$. 
\begin{figure}[htb]
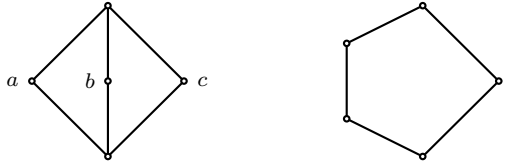

\begin{center}
\psset{unit=0.5cm}
\pspicture(0,0)(4,4)
\pspolygon(2,0)(0,2)(2,4)(4,2)
\psline(2,0)(2,4)
\pscircle[fillstyle=solid](2,0){0.1}
\pscircle[fillstyle=solid](0,2){0.1}
\pscircle[fillstyle=solid](2,4){0.1}
\pscircle[fillstyle=solid](4,2){0.1}
\pscircle[fillstyle=solid](2,2){0.1}
\uput[180](0,2){\scriptsize  $a$}
\uput[180](2,2){\scriptsize  $b$}
\uput[0](4,2){\scriptsize  $c$}
\endpspicture
\hspace*{2cm}
\pspicture(0,0)(4,4)
\pspolygon(2,0)(0,1)(0,3)(2,4)(4,2)
\pscircle[fillstyle=solid](2,0){0.1}
\pscircle[fillstyle=solid](0,1){0.1}
\pscircle[fillstyle=solid](0,3){0.1}
\pscircle[fillstyle=solid](2,4){0.1}
\pscircle[fillstyle=solid](4,2){0.1}
\endpspicture
\end{center}
\caption{The lattices $M_3$ (left) and $N_5$ (right)}
\label{fig:M3}
\end{figure}
$(L,\leq)$ is said to be \emph{lower (upper) locally distributive} if it is
lower (upper) semimodular, and it does not contain a sublattice isomorphic to
$M_3$.  These are weaker conditions than distributivity, and if $L$ is both
lower and upper locally distributive, then it is distributive.

An element $j\in L$ is \emph{join-irreducible} if $j=x\vee y$ implies either
$j=x$ or $j=y$, i.e., it cannot be expressed as a supremum of other
elements. Equivalently $j$ is join-irreducible if it covers only one element.
Join-irreducible elements covering $\bot$ are called \emph{atoms}, and the
lattice is \emph{atomistic} if all join-irreducible elements are atoms. The set
of all join-irreducible elements of $L$ is denoted $\mathcal{J}(L)$. On
Fig. \ref{fig:lat1} and \ref{fig:lat2}, they are figured as black nodes.

Similarly, \emph{meet-irreducible elements} cannot be written as an infimum of
other elements, and are such that they are covered by a single
element. We denote by $\mathcal{M}(L)$ the set of meet-irreducible elements of
$L$. \emph{Co-atoms} are meet-irreducible elements covered by $\top$.

For any $x\in L$, we say that $x$ \emph{has a complement in $L$} if there exists
$x'\in L$ such that $x\wedge x'=\bot$ and $x\vee x'=\top$. The complement is
unique if the lattice is distributive. $L$ is said to be \emph{complemented} if
any element has a complement. On Fig. \ref{fig:lat1} (left), no element has a
complement, except top and bottom, while the two others are complemented
lattices.

\emph{Boolean lattices} are distributive and complemented lattices, and in a
finite setting, they are of the type $2^N$ for some set $N$, i.e. they are
isomorphic to the lattice of subsets of some set, ordered by inclusion (see
Fig. \ref{fig:lat1} (middle,right)). Boolean lattices are atomistic, and atoms
correspond to singletons, while co-atoms are of the form $N\setminus\{i\}$ for
some $i\in N$.

\medskip

An important property is that in a lower locally distributive lattice, any
element $x$ can be written as an irredundant supremum of join-irreducible
elements in a unique way (this is called the \emph{minimal decomposition} of
$x$). We denote by $\eta^*(x)$ the set of join-irreducible elements in the
minimal decomposition of $x$, and we denote by $\eta(x)$ the \emph{normal
decomposition} of $x$, defined as the set of join-irreducible elements smaller
or equal to $x$, i.e., $\eta(x):=\{j\in \mathcal{J}(L)\mid j\leq x\}$.  Hence
$\eta^*(x)\subseteq \eta(x)$, and
\[
x=\bigvee_{j\in \eta^*(x)} j = \bigvee_{j\in \eta(x)} j.
\]
Put differently, the mapping $\eta$ is an isomorphism of
$L$ onto $\mathcal{O}(\mathcal{J}(L))$ (Birkhoff's theorem).

Likewise, any element in a upper locally distributive lattice can be written as
a unique irredundant infimum of meet-irreducible elements. The decomposition
are denoted by $\mu$ and $\mu^*$. Specifically,
$\mu(x):=\{m\in\mathcal{M}(L)\mid m\geq x\}$, and
$\displaystyle{x=\bigwedge_{m\in\mu(x)}m}$.

\medskip

The \emph{height function} $h$ on $L$ gives the length of a longest chain from
$\bot$ to any element in $L$. A lattice is \emph{ranked} if $x\succ y$ implies
$h(x)=h(y)+1$. A lattice is lower locally distributive if and only if it is
ranked and the length of any maximal chain is $|\mathcal{J}(L)|$.

\subsection{The M\"obius and co-M\"obius transforms}
\label{sec:mobi}
We follow the general definition of Rota \cite{rot64} (see also \cite[p.
102]{bir67}). Let $(L,\leq)$ be a poset which is locally finite
(i.e., any interval is finite) having a bottom element. For any function $f$ on
$(L,\leq)$, the \emph{M\"obius transform} of $f$ is the function
$m:L\longrightarrow \mathbb{R}$ solution of the equation:
\begin{equation}\label{eq:mob}
f(x) = \sum_{y\leq x}m(y).
\end{equation}
This equation has always a unique solution, and the expression of $m$ is
obtained through the M\"obius function $\mu:L^2\rightarrow \mathbb{R}$ by:
\begin{equation}\label{eq:invm}
m(x) = \sum_{y\leq x}\mu(y,x)f(y)
\end{equation}
where $\mu$ is defined inductively by
\begin{equation}\label{eq:mu}
\mu(x,y) = \left\{      \begin{array}{ll}
                        1, & \text{ if } x=y\\
                        -\sum_{x\leq t< y}\mu(x,t), & \text{ if } x< y\\
                        0, &  \text{ otherwise}.
                        \end{array}     \right.
\end{equation}
Note that $\mu$ depends solely on $L$.

The \emph{co-M\"obius transform} of $f$, denoted by $q$, is defined by
\cite{gra98c,grla00a}:
\begin{equation}\label{eq:com}
q(x):=\sum_{y\geq x} m(y), \quad x\in L.
\end{equation}

\subsection{Belief functions and related concepts}
\label{sec:beli}
We recall only necessary definitions. For details, the reader is referred to,
e.g., \cite{sme90,sme94}, or the monograph \cite{kra01a}.

Let $\Omega$ be a finite space. A function $m:2^\Omega\rightarrow [0,1]$ is said
to be a \emph{mass allocation function} (or simply a \emph{mass}) if
$m(\emptyset)=0$ and $\sum_{A\subseteq \Omega}m(A)=1$. A subset $A\subseteq N$
is said to be a \emph{focal element} if $m(A)>0$.

A \emph{belief function} on $\Omega$ is a function $\bel:2^\Omega\rightarrow
[0,1]$ generated by a mass allocation function as follows:
\begin{equation}\label{eq:bel}
\bel(A) := \sum_{B\subseteq A} m(B), \quad A\subseteq \Omega.
\end{equation} 
Note that $\bel(\emptyset)=0$ and $\bel(\Omega)=1$. One recognizes $m$ as being
the M\"obius transform of $\bel$ (apply Eq. (\ref{eq:mob}) to
$(L,\leq):=(2^\Omega,\subseteq)$). The inverse formula, obtained by using
(\ref{eq:invm}) and (\ref{eq:mu}), is:
\begin{equation}\label{eq:invb}
m(A) = \sum_{B\subseteq A} (-1)^{|A\setminus B|}\bel(B).
\end{equation}
Given a mass allocation $m$, the \emph{plausibility function} is defined by:
\begin{equation}\label{eq:plau}
\pl(A):=\sum_{B\mid A\cap B\neq\emptyset} m(B) = 1-\bel(A^c), \quad A\subseteq \Omega.
\end{equation}
Similarly, the \emph{commonality function} is defined by:
\begin{equation}\label{eq:comm}
q(A):=\sum_{B\supseteq A} m(B) , \quad A\subseteq \Omega.
\end{equation}
It is the co-M\"obius transform of $\bel$ (see (\ref{eq:com})). Remark that
$q(\emptyset)=1$.  

A \emph{capacity} on $\Omega$ is a set function $v:2^\Omega\rightarrow [0,1]$
such that $v(\emptyset)=0$, $v(\Omega)=1$, and $A\subseteq B$ implies
$v(A)\leq v(B)$ (\emph{monotonicity}). Plausibility and belief functions are
capacities. For any capacity $v$, its \emph{conjugate} is defined by
$\overline{v}(A):=1-v(A^c)$. Hence, plausibility functions are conjugate of
belief functions (and vice versa). A capacity is \emph{$k$-monotone}
($k\geq 2$) if for any family of $k$ subsets $A_1,\ldots,A_k$ of $\Omega$, it
holds:
\begin{equation}
v(\bigcup_{i\in K} A_i)\geq \sum_{I\subseteq
  K,I\neq\emptyset}(-1)^{|I|+1}v(\bigcap_{i\in I}A_i),
\end{equation}
with $K:=\{1,\ldots,k\}$. A capacity is \emph{totally monotone} if it is
$k$-monotone for every $k\geq 2$. 

Shafer \cite{sha76} has shown that a capacity is totally monotone if and only if
it is a belief function, hence there exists some mass allocation generating it.

Given two mass allocations $m_1,m_2$, the \emph{Dempster's rule of combination}
computes a combination of both masses into a single one:
\begin{equation}
m(A) =:(m_1\oplus m_2)(A):= \sum_{B_1\cap B_2=A}m_1(B_1)m_2(B_2), \quad
\forall A\subseteq \Omega, A\neq\emptyset,
\end{equation}
and $m(\emptyset):=0$.  Note that $m$ is no more a mass allocation in general,
unless some normalization is carried out. It is well known that the Dempster
rule of combination can be computed through the commonality functions much more
easily. Specifically, calling $q,q_1,q_2$ the commonality functions associated
to $m,m_1,m_2$, one has:
\begin{equation}\label{eq:drq}
q(A) = q_1(A)q_2(A), \quad \forall A\subseteq \Omega.
\end{equation}

A \emph{simple support function focused on $A$} is a particular belief function
$\bel_A$ whose mass allocation is:
\begin{equation}\label{eq:ssf}
m_A(B):=\begin{cases}
  1 - w_A, & \text{if } B=A\\
  w_A, & \text{if } B=\Omega\\
  0, & \text{otherwise}.
  \end{cases}
\end{equation}
with $0<w_A<1$. Smets \cite{sme95}, using results of Shafer, has shown that
any belief function such that $m(\Omega)\neq 0$ can be decomposed using only
simple support functions as follows:
\begin{equation}\label{eq:decom}
\bel=\bigoplus_{A\subseteq \Omega}\bel_A
\end{equation}
with
\begin{equation}\label{eq:decom1}
w_A=\prod_{B\supseteq A}q(B)^{(-1)^{|B\setminus A|+1}}, \quad\forall A\subseteq
\Omega. 
\end{equation}
In the above decomposition, coefficients $w_A$ may be greater than 1. If this
happens, the corresponding $\bel_A$ is no more a belief function.

\medskip

A \emph{necessity function} or \emph{necessity measure} is a belief function
whose focal elements form a chain in $(2^\Omega,\subseteq)$, i.e.,
$A_1\subseteq A_2\subseteq\cdots\subseteq A_n$ (Dubois and Prade,
\cite{dupr85b}). The characteristic property of necessity functions is that for
any subsets $A,B$, $\N(A\cap B)=\min(\N(A),\N(B))$, where $\N$ denotes a
necessity function.

Conjugates of necessity functions are called \emph{possibility functions},
denoted by $\Pi$, and are particular plausibility functions. It is easy to see
that their characteristic property is that for any subsets $A,B$, $\Pi(A\cup
B)=\max(\Pi(A),\Pi(B))$. This characteristic property implies that $\Pi$ is
entirely determined by its value on singletons, i.e., $\Pi(A)=\max_{\omega\in
A}\Pi(\{\omega\})$ for any $A\subseteq \Omega$. For this reason,
$\pi(\omega):=\Pi(\{\omega\})$ is called the \emph{possibility distribution}
associated to $\Pi$. Note that necessarily there exists $\omega_0\in \Omega$
such that $\pi(\omega_0)=1$. Although this is generally not considered, one may
define as well a \emph{necessity distribution}
$\nu(\omega):=\N(\Omega\setminus\omega)$, with the property that
$\N(A)=\min_{\omega\in A^c}\nu(\omega)$. 

Let $\pi$ be a possibility distribution on
$\Omega:=\{\omega_1,\ldots,\omega_n\}$, and assume that for some permutation
$\sigma$ on $\{1.\ldots,n\}$, it holds $\pi(\omega_{\sigma(1)})
\leq\pi(\omega_{\sigma(2)})\leq \cdots\leq \pi(\omega_{\sigma(n)})=1$. Then it
can be shown that the focal elements of the mass allocation associated to $\Pi$
are of the form
$A_{\sigma(i)}:=\{\omega_{\sigma(i)},\ldots,\omega_{\sigma(n)}\}$, $i=1,\ldots,
n$, and $m(A_{\sigma(i)})=\pi(\omega_{\sigma(i)}) - \pi(\omega_{\sigma(i-1)})$,
with the convention $\pi(\omega_{\sigma(0)})=0$. 

\section{Belief functions and capacities on lattices}
\label{sec:bella}
Let $(L,\leq)$ be a finite lattice. A \emph{capacity} on $L$ is a function
$v:L\rightarrow [0,1]$ such that $v(\bot)=0$, $v(\top)=1$, and $x\leq y$
implies $v(x)\leq v(y)$ (isotonicity). 

To define the conjugate of a capacity, a natural way would be to write
$\overline{v}(x):=1-v(x')$, where $x'$ is the complement of $x$. But this
would impose that $L$ is complemented, which is very restrictive. For example,
the lattice $3^n$ underlying bi-belief functions is not complemented. Moreover,
if distributivity is imposed in addition, then only Boolean lattices are
allowed, and we are back to the classical definition. We adopt a more general
definition.
\begin{definition}
A lattice $L$ is of \emph{De Morgan type} if it exists a bijective mapping
$n:L\rightarrow L$ such that for any $x,y\in L$ it holds $n(x\vee y)=n(x)\wedge
n(y)$, and $n(\top)=\bot$. We call such a mapping a \emph{$\vee$-negation}.
\end{definition} 
The following is immediate.
\begin{lemma}\label{lem:nega}
Let $L$ be a De Morgan lattice, with $n$ a $\vee$-negation. Then:
\begin{itemize}
\item [(i)] $n(\bot)=\top$.
\item [(ii)] $n^{-1}(x\wedge y)=n^{-1}(x)\vee
n^{-1}(y)$, for all $x,y\in L$ ($n^{-1}$ is called a \emph{$\wedge$-negation}).
\item [(iii)] If $j$ is join-irreducible, then
  $n(j)$ is meet-irreducible, and if $m$ is meet-irreducible, then $n^{-1}(m)$ is
join-irreducible. 
\end{itemize}
\end{lemma}
\begin{proof}
(i) $n(x\vee \bot)=n(x) =
  n(x)\wedge n(\bot)$, for all $x\in L$, which implies
  $n(\bot)=\top$ because $n$ is a bijection.

(ii) Putting
$x':=n(x)$ and $y':=n(y)$, we have $n^{-1}(x'\wedge y')=n^{-1}(n(x\vee
y))=x\vee y= n^{-1}(x')\vee n^{-1}(y')$.

(iii) If $j$ is join-irreducible, $j=x\vee y$ implies that $j=x$ or
  $j=y$. Hence, $n(j)=n(x\vee y)=n(x)\wedge n(y)$ is either $n(x)$ or $n(y)$,
  which means that $n(j)$ is meet-irreducible.
\end{proof}

A complemented lattice with unique complement is of De Morgan type with
$n(x):=x'$. If $L$ is isomorphic to its dual $L^\partial$, i.e. it is autodual,
then it is of De Morgan type since it suffices to take for $n(x)$ the element
in the Hasse diagram of $L^\partial$ which takes the place of $x$ in the Hasse
diagram of $L$. In this case, we call $n$ a \emph{horizontal symmetry}. In
general, $n$ is not unique since there is no unique way to draw Hasse
diagrams. Taking lattices of Fig. \ref{fig:lat1} as examples, for the left one,
we would have $n(a)=e$, for the middle one $n(12)=1$, and for the right one
$n(12)=3$ (see Fig. \ref{fig:lat2}). Since middle and right lattices are the
same, this shows that several $n$ exist in general. Note that $n$ for the right
lattice is nothing else than the usual complement.

The following result shows that in fact the only De Morgan type lattices are
those which are autodual.  
\begin{proposition}\label{prop:gal}
A lattice $L$ is of De Morgan type if and only if it is autodual.
\end{proposition}
\begin{proof}
We already know that if $L$ is autodual, then it is of De Morgan type.
Conversely, assuming it is of De Morgan type, it suffices to show that $n$ is an
anti-isomorphism. We already know that $n$ is a bijection. Taking $x\leq y$
implies that $x\vee y=y$, hence $n(x\vee y)=n(y)=n(x)\wedge n(y)$, which implies
$n(y)\leq n(x)$. Conversely, $n(y)\leq n(x)$ implies $n(y)\wedge
n(x)=n(y)=n(x\vee y)$, hence $x\vee y=y$ since $n$ is a bijection, so that
$x\leq y$. 
\end{proof}

In general, $n$ and $n^{-1}$ differ, that is, $n$ is not always
\emph{involutive}. Take for example the lattice $M_3$ of Fig. \ref{fig:M3}, and
$n$ defined by $n(\top)=\bot$, $n(\bot)=\top$, $n(a)=b$, $n(b)=c$ and
$n(c)=a$. Clearly, $n$ is a $\vee$-negation, but $n(n(a))=c\neq a$. The
$\vee$-negation is involutive whenever $n$ is a horizontal symmetry on the
Hasse diagram. If $n$ is involutive, it is simply called a \emph{negation}.

\begin{definition}
Let $L$ be an autodual lattice, and $n$ a $\vee$-negation on $L$. For any
capacity $v$, its \emph{$\vee$-conjugate} and \emph{$\wedge$-conjugate}
(w.r.t. $n$) are defined respectively by 
\begin{align*}
{}^\vee\overline{v}(x) &:=1-v(n(x))\\
{}^\wedge\overline{v}(x) &:=1-v(n^{-1}(x)),
\end{align*}
for any $x\in L$. If $n$ is a negation, then $\overline{v}(x):=1-v(n(x))$ is
the \emph{conjugate} of $v$.
\end{definition}
The following is immediate.
\begin{lemma}\label{lem:neg}
Let $L$ be an autodual lattice, and $n$ a $\vee$-negation on $L$. For any
capacity $v$, it holds
\begin{itemize}
\item [(i)] ${}^\vee\overline{v}$ and ${}^\wedge\overline{v}$ are capacities on
$L$.
\item [(ii)] ${}^\vee\overline{{}^\wedge\overline{v}}
={}^\wedge\overline{{}^\vee\overline{v}}=v$.  
\end{itemize}
\end{lemma}
\begin{proof}
(i) ${}^\vee\overline{v}(\top)=1-v(n(\top))=1$, similarly for
$\bot$. Isotonicity of ${}^\vee\overline{v}$ follows from antitonicity of $n$
and isotonicity of $v$. 

(ii) ${}^\vee\overline{{}^\wedge\overline{v}}(x)=1-{}^\wedge\overline{v}(n(x))
= 1-(1-v(x))=v(x)$. 
\end{proof}

\medskip

The following definition of belief functions is in the spirit of the original
one by Shafer. We used it also for defining bi-belief functions \cite{grla03}.
\begin{definition}
  A function $\bel:L\rightarrow [0,1]$ is called a \emph{belief function} if
  $\bel(\top)=1$, $\bel(\bot)=0$, and its M\"obius transform is non negative.
\end{definition}
Referring to (\ref{eq:mob}), we recall that
\begin{equation}\label{eq:bell}
\bel(x)=\sum_{y\leq x}m(y),\quad \forall x\in L.
\end{equation} 
Note that $\bel(\top)=1$ is equivalent to $\sum_{x\in L}m(x)=1$, and
$\bel(\bot)=0$ is equivalent to $m(\bot)=0$. The inverse formula, giving $m$ in
terms of $\bel$, has to be computed from (\ref{eq:mu}), and depends only on the
structure of $L$. 

Remark that $\bel$ is an isotone function by nonnegativity of $m$, and hence a
capacity.

Thanks to the definition of conjugation, if $L$ is autodual and $n$ is a
$\vee$-negation, one can define \emph{plausibility functions} as the
$\vee$-conjugate of belief functions, which are again capacities.

\subsection{$k$-monotone functions}
\label{sec:kmon}
Barth\'elemy defines belief function as totally monotone functions. To detail
this point, we define $k$-monotone functions. For $k\geq 2$, a function
$f:L\rightarrow \mathbb{R}$ is said to be \emph{$k$-monotone} (called
\emph{weakly $k$-monotone} by Barth\'elemy) if it satisfies, for any family of
elements $x_1,\ldots,x_k\in L$:
\begin{equation}\label{eq:kmon}
f(\bigvee_{i\in K}x_i) \geq \sum_{I\subseteq K, I\neq\emptyset}(-1)^{|I|+1}
f(\bigwedge_{i\in I}x_i)
\end{equation}
where $K:=\{1,\ldots,k\}$. 
A function is said to be \emph{totally monotone} if it is $k$-monotone for all
$k\geq 2$. One can prove that in fact, if $|L|=n$, total monotonicity is
equivalent to $(n-2)$-monotonicity \cite{bar00}.

For $k\geq 2$, a function is said to be a \emph{$k$-valuation} if the inequality
(\ref{eq:kmon}) degenerates into an equality (called also \emph{Poincar\'e's
  inequality}). Similarly, a function is an \emph{infinite valuation} or
\emph{total valuation} if it is a $k$-valuation for all $k\geq 2$. It is well
known that monotone infinite valuations satisfying $f(\top)=1$ and $f(\bot)=0$
are probability measures.

The following lemma, cited in \cite{bar00}, summarizes well-known results from
lattice theory (see Birkhoff \cite{bir67}).
\begin{lemma}
Let $L$ be a lattice. Then
\begin{itemize}
\item [(i)] $L$ is modular if and only if it admits a strictly monotone
  2-valuation.
\item [(ii)] $L$ is distributive if and only if it is modular and every strictly
  monotone 2-valuation on $L$ is a 3-valuation. 
\item [(iii)] $L$ is distributive if and only if it admits a strictly monotone
  3-valuation.
\item [(iv)] $L$ is distributive if and only if it is modular and every strictly
  monotone 2-valuation on $L$ is an infinite valuation.
\end{itemize}
\end{lemma}  
Barth\'elemy showed in addition that any lattice admits a totally monotone
function. In view of this result, Barth\'elemy defines belief functions as totally
monotone function being monotone and satisfying $f(\top)=1$ and $f(\bot)=0$. In
summary, a belief function can be defined on any lattice, while probability
measures can live only on distributive lattices.

The following proposition shows the relation between both definitions. Before,
we state a result from \cite{bar00}.
\begin{lemma}\label{lem:2}
  For any lattice $L$ and any function $m:L\rightarrow [0,1]$ such that
  $m(\bot)=0$ and $\sum_{x\in L}m(x)=1$, the function $f^m:L\rightarrow [0,1]$
  defined by $f^m(x):=\sum_{y\leq x}m(y)$ is totally monotone and satisfies
  $f^m(\top)=1$ and $f^m(\bot)=0$. 
\end{lemma}
\begin{proposition}
Any belief function is totally monotone.
\end{proposition} 
\begin{proof}
Let $\bel$ be a belief function, and $m$ its M\"obius transform. We know that
$m(\bot)=0$ and $\sum_{x\in L}m(x)=1$. Hence, by Lemma \ref{lem:2}, $\bel$ is
totally monotone.
\end{proof}

A totally monotone function does not have necessarily a non negative M\"obius
function. Simple examples show that monotonicity is a necessary condition. The
question to know whether monotonicity and total monotonicity imply non
negativity of the M\"obius function is still open. 

\subsection{Properties of belief functions}
A first result shown by Barth\'elemy shows that capacities collapse to belief
functions when $L$ is linear \cite{bar00}.
\begin{proposition}
Any capacity on $L$ is a belief function if and only if $L$ is a linear lattice.
\end{proposition}

In the sequel, we address the combination of belief functions and their
decomposition in terms of simple support functions. We will see that classical
results generalize. 
\begin{definition}\label{def:drc}
  Let $\bel_1,\bel_2$ be two belief functions on $L$, with M\"obius transforms
  $m_1,m_2$. The \emph{Dempster's rule of combination} of $\bel_1,\bel_2$ is
  defined through its M\"obius transform $m$ by:
\[
m(x)=:(m_1\oplus m_2)(x):=\sum_{y_1\wedge y_2=x}m_1(y_1)m_2(y_2), \quad\forall
x\in L.
\]
\end{definition}
Since $m$ defines unambiguously the belief function, we may write as well
$\bel=\bel_1\oplus\bel_2$ to denote the combination.

\begin{proposition}\label{prop:comb}
Let $\bel_1,\bel_2$ be two belief functions on $L$, with co-M\"obius transforms
$q_1,q_2$, and consider their Dempster combination. Then, if $q$ denotes the
co-M\"obius transform of $\bel:= \bel_1\oplus\bel_2$,
\[
q(x)= q_1(x)q_2(x),\quad \forall x\in L.
\] 
\end{proposition}
\begin{proof}
We have:
\[
q(x) = \sum_{y\geq x}\sum_{y_1\wedge y_2=y}m_1(y_1)m_2(y_2) = \sum_{y_1\wedge
  y_2\geq x} m_1(y_1)m_2(y_2).
\]
One can decompose the above sum since if $y_1\geq x$ and $y_2\geq x$, then
$y_1\wedge y_2\geq x$ and reciprocally. Thus,
\[
q(x) = \sum_{y_1\geq x}m(y_1)\sum_{y_2\geq x}m(y_2) = q_1(x)q_2(x).
\]
\end{proof}
The above proposition generalizes (\ref{eq:drq}), and gives a simple means to
compute the Dempster combination.

\begin{remark}
  In Def. \ref{def:drc}, one may put as in the classical case $m(\bot)=0$. This
  does not affect the validity of Prop. \ref{prop:comb}, except for $x=\bot$.
  Indeed, by Prop. \ref{prop:comb}, one obtains $q(\bot)=1$, but
  $q(\bot)=\sum_{x\in L}m(x)<1$ in general if one puts $m(\bot)=0$ in Def.
  \ref{def:drc}.
\end{remark}

\begin{definition}
Let $y\in L$. A \emph{simple support function focused on $y$}, denoted by
$y^w$, is a function on $L$ such that its M\"obius transform satisfies:
\[
m(x) = \begin{cases}
  1-w, & \text{if } x=y\\
  w, & \text{if } x=\top\\
  0, & \text{otherwise. }
  \end{cases}
\] 
\end{definition}

The decomposition of some belief function $\bel$ in terms of simple support
functions is thus to write $\bel$ under the form:
\begin{equation}\label{eq:beldec}
\bel(x) = \bigoplus_{y\in L} y^{w_y}(x).
\end{equation}
The following result generalizes the decomposition in the classical case (see
Sec. \ref{sec:beli}).
\begin{theorem}
Let $\bel$ be a belief function such that its M\"obius transform $m$ satisfies
$m(\top)\neq 0$. 
The coefficients $w_y$ of the decomposition (\ref{eq:beldec}) write
\[
w_y=\prod_{x\geq y}q(x)^{-\mu(x,y)}
\]
where $\mu(x,y)$ is the M\"obius function of $L$.
\end{theorem}
\begin{proof}
We try to find $w_y$ such that
\[
\bel(x) = \bigoplus_{y\in L}y^{w_y}.
\]
This expression can be written in terms of the co-M\"obius transform:
\begin{equation}\label{eq:q}
q(x)= \prod_{y\in L}q_y(x),\quad x\in L,
\end{equation}
where $q_y$ is the co-M\"obius transform of $y^{w_y}$:
\[
q_y(x)=\begin{cases}
  1, & \text{if } x\leq y\\
  w_y, & \text{otherwise.}
  \end{cases}
\]
From (\ref{eq:q}), we obtain:
\begin{align*}
\log q(x) & = \sum_{y\in L}\log q_y(x) = \sum_{y\not\geq x}\log w_y \\
 & = \sum_{y\in L}\log w_y - \sum_{y\geq x}\log w_y.
\end{align*}
On the other hand,
\[
q(\top) = \prod_{y\in L}q_y(\top) = \prod_{y\in L}w_y.
\]
We supposed that $q(\top)=m(\top)\neq 0$, hence:
\[
\log q(x) = \log q(\top) - \sum_{y\geq x}\log w_y.
\]
We set $Q(x):=\log q(x)$ and $W(y):=\log w_y$. The last equality becomes:
\[
Q(x) = Q(\top)-\sum_{y\geq x}W(y).
\]
If we define $Q'(x) = Q(\top) - Q(x)$, we finally obtain:
\[
Q'(x) = \sum_{y\geq x} W(y).
\]
We recognize here the equation defining the M\"obius transform of $Q'$, up to an
inversion of the order (dual order)(see (\ref{eq:mob})). Hence, using
(\ref{eq:invm}):
\[
W(y)=\sum_{x\geq y}\mu(x,y)Q'(x)
\]
with $\mu$ defined by (\ref{eq:mu}). Rewriting this with original notation, we
obtain:
\[
\log w_y=\sum_{x\geq y}\mu(x,y)[\log q(\top) - \log q(x)].
\] 
Remarking that $\sum_{x\geq y}\mu(x,y)\log q(\top)$ is zero, since it
corresponds to the M\"obius transform of a constant function, we finally get:
\[
w_y=\prod_{x\geq y} q(x)^{-\mu(x,y)}.
\]
\end{proof}
Note that the above proof is much shorter and general than the original one by
Shafer \cite{sha76}.

As in the classical case, these coefficients may be strictly greater
than 1, hence corresponding simple support functions have negative M\"obius
transform and are no more belief functions.

\section{Necessity functions}
\label{sec:nece}
\begin{definition}
  A function $\N:L\rightarrow[0,1]$ is called a \emph{necessity function} if it
  satisfies $\N(x\wedge y)=\min(\N(x),\N(y))$, for all $x,y\in L$, and
  $\N(\bot)=0$, $\N(\top)=1$.
\end{definition}
The following result is due to Barth\'elemy \cite{bar00}.
\begin{proposition}
$\N$ is a necessity function if and only if it is belief function whose M\"obius
transform $m$ is such that its focal elements form a chain in $L$. 
\end{proposition}

We define possibility functions as $\vee$-conjugates of necessity functions.
\begin{definition}
  Let $L$ be an autodual lattice, and $n$ a $\vee$-negation on $L$. For any
  necessity function $\N$ on $L$, its $\vee$-conjugate is called a
  \emph{possibility function}.
\end{definition}
Let $\Pi$ be a possibility function. Then ${}^\wedge\overline{\Pi}$ is its
corresponding necessity function by Lemma \ref{lem:neg} (ii).
\begin{proposition}
Let $L$ be an autodual lattice, and $n$ a $\vee$-negation on $L$. The mapping
$\Pi:L\rightarrow [0,1]$ is a possibility function if and only if
\begin{equation}\label{eq:pi}
\Pi(x\vee y) = \max(\Pi(x),\Pi(y)), \quad \forall x,y\in L.
\end{equation}
\end{proposition}
\begin{proof}
  Let $\Pi$ be a possibility function being the $\vee$-conjugate of some
  necessity function $N$. Then:
\begin{align*}
\Pi(x\vee y) & = 1-\N(n(x\vee y)) = 1- \N(n(x)\wedge n(y))\\
             & = 1- \min(\N(n(x)),\N(n(y))) =
             \max(1-\N(n(x)),1-\N(n(y)))\\ 
             & = \max(\Pi(x),\Pi(y)).
\end{align*} 
Conversely, let $\Pi$ satisfy (\ref{eq:pi}) and consider its $\wedge$-conjugate
${}^\wedge\overline{\Pi}$. We have:
\begin{align*}
{}^\wedge\overline{\Pi}(x\wedge y) & = 1-\Pi(n^{-1}(x\wedge y))=
                          1-\Pi(n^{-1}(x)\vee n^{-1}(y))\\ 
        & = 1-\max(\Pi(n^{-1}(x)),\Pi(n^{-1}(y)))\\ 
        & = \min({}^\wedge\overline{\Pi}(x),{}^\wedge\overline{\Pi}(y)).
\end{align*}
Hence ${}^\wedge\overline{\Pi}$ is a necessity function, which implies that
$\Pi$ is a possibility function since
${}^\vee\overline{{}^\wedge\overline{\Pi}}=\Pi$ by Lemma \ref{lem:neg} (ii).
\end{proof}

The next topic we address concerns distributions. Since we need the property of
decomposition of elements into supremum of join-irreducible elements, we impose
that $L$ is lower locally distributive. Since $L$ has to be autodual, then it is
also upper locally distributive, and so it is distributive. We propose the
following definition.
\begin{definition}
  Let $L$ be an autodual distributive lattice, some
  $\vee$-negation $n$ on $L$, and $\N$ a necessity function. The
  \emph{possibility distribution} $\pi:\mathcal{J}(L)\rightarrow [0,1]$
  associated to $\N$ is defined by $\pi(j):=\Pi(\{j\})$, $j\in \mathcal{J}(L)$,
  with $\Pi$ the possibility function which is $\vee$-conjugate of $\N$.

The \emph{necessity distribution} $\nu:\mathcal{M}(L)\rightarrow [0,1]$
associated to $\N$ is defined by $\nu(m):=\N(\{m\})$,
  $m\in\mathcal{M}(L)$.  
\end{definition}
Then, the value of $\Pi$ and $\N$ at any $x\in L$ can be computed as follows:
\[
\Pi(x) = \max(\pi(j) \mid j\in\eta^*(x)), \quad \N(x) = \min(\nu(m)\mid
m\in\mu^*(x)).  
\]
Remark that due to isotonicity of $\Pi$ and $\N$, and hence of $\pi$ and $\nu$,
one can replace as well $\eta^*,\mu^*$ by $\eta,\mu$. The above formulas are
well-defined since the decomposition is unique for distributive
lattices. Lastly, remark that necessarily there exists $j_0\in\mathcal{J}(L)$
such that $\pi(j_0)=1$, and $m_0\in\mathcal{M}(L)$ such that $\nu(m_0)=0$,
since $\Pi(\top)=1$ and $\N(\bot)=0$.

$\pi$ and $\nu$ are related through conjugation since $n$ maps join-irreducible
elements to meet-irreducible elements and vice-versa for $n^{-1}$ (see Lemma
\ref{lem:nega} (iii)). Hence, for $j\in\mathcal{J}(L)$ and
$m\in\mathcal{M}(L)$:
\begin{align*}
\pi(j) & =1 - \N(n(j)) = 1 - \nu(m_j)\\
\nu(m) & =1 - \Pi(n^{-1}(m)) = 1 - \pi(j_m),
\end{align*}
where $m_j:=n(j)$, and $j_m:=n^{-1}(m)$.

\medskip

Given a mass allocation defining some necessity function, it is easy to derive
the corresponding possibility distribution. The converse problem, i.e., given a
possibility distribution, find (if possible) the corresponding chain of focal
elements and mass allocation giving rise to this possibility distribution, is
less simple. Interestingly enough, this problem has always a unique solution,
which is very close to the classical case.
\begin{theorem}\label{th:2}
  Let $L$ be autodual, distributive, and $n$ be a $\vee$-negation on $L$. Let
  $\pi$ be a possibility distribution, and assume that the join-irreducible
  elements of $L$ are numbered such that $\pi(j_1)<\cdots <\pi(j_n)=1$.
  Then there is a unique maximal chain of focal elements generating $\pi$, given
  by the following procedure:
\begin{quote}
  Going from $j_n$ to $j_1$, at each step $k=n,n-1,\ldots,1$, select the unique
  join-irreducible element $\iota_k$ such that:
\begin{equation}\label{eq:proc}
\iota_k\not\in \eta(n(j_k)),\quad \iota_k\in\bigcap_{l=1}^{k-1}\eta(n(j_l)).
\end{equation}
Then the maximal chain is defined by $C_\pi:=\{\iota_n,
\iota_n\vee\iota_{n-1},\ldots, \iota_n\vee\cdots\vee\iota_2,\top\}$, and
\begin{equation}\label{eq:mobpi}
m(\iota_n\vee\iota_{n-1}\vee\cdots\vee\iota_k)=\pi(j_k)-\pi(j_{k-1}),\quad k=1,\ldots,n,
\end{equation}
with $\pi(j_0):=0$.
Moreover, at each step $k$,  it is equivalent to choose $\iota_k$ as the
smallest in $\eta(n(j_{k-1}))\setminus \eta(n(j_k))$.
\end{quote}
\end{theorem}
\begin{proof}
For ease of notation, denote $n(j_k)$ by $m_k$ (meet-irreducible). 

We first show that such a procedure can always work and leads to a unique
solution for $C_\pi$. Assume that the poset $\mathcal{J}(L)$ has $q$ connected
components $J_1,\ldots,J_q$. By definition, $j_n$ is one of the maximal
elements of one of the connected components, say $J_{q_0}$. Clearly,
$\bigwedge_{k=1}^n m_k=n(\bigvee_{k=1}^n j_k)=n(\top)=\bot$. But
$\bigvee_{k=1}^{n-1} j_k\neq\top$, otherwise $\big(\bigcup_{l=1,\ldots,q,l\neq
q_0}J_l\big)\cup J'_{q_0}$, where $J'_{q_0}$ is a maximal downset of
$J_{q_0}\setminus \{j_n\}$, would be another downset corresponding to $\top$,
which is impossible since $L$ is distributive (Birkhoff's theorem). This
implies that there exists $\iota_n\in \bigcap_{k=1}^{n-1}\eta(m_k)$, and
$\iota_n\not\in \eta(m_n)$.  Let us show that $\iota_n$ is unique. Since $L$ is
distributive, it is ranked and any maximal chain has length
$|\mathcal{J}(L)|=n$. Hence, $\bigvee_{k=1}^{n-1}j_k$ has height $n-1$ (it is a
co-atom), and $\bigwedge_{k=1}^{n-1}m_k$ is an atom.  Therefore,
$\bigcap_{k=1}^{n-1}\eta(m_k)$ is a singleton.

For $\iota_{n-1}$ and subsequent ones, we apply the same reasoning on the
lattice $\mathcal{O}(\mathcal{J}(L)\setminus\{\iota_n\})$, then on
$\mathcal{O}(\mathcal{J}(L)\setminus\{\iota_n,\iota_{n-1}\})$, etc., instead of
$L=\mathcal{O}(\mathcal{J}(L))$. Hence, there will be $n$ steps, and at each
step one join-irreducible element is chosen in a unique way.

We prove now that the sequence $\{\iota_n, \iota_n\vee\iota_{n-1},\ldots,
\iota_n\vee\cdots\vee\iota_2,\top\}$ is a maximal chain, denoted $C_\pi$. It
suffices to prove that $\iota_n\vee\iota_{n-1}\vee\cdots\vee\iota_k\succ
\iota_n\vee\iota_{n-1}\vee\cdots\vee\iota_{k+1}$, $k=1,\ldots,n-1$. The fact
that the former is greater or equal to the latter is obvious, hence $C_\pi$ is
a chain. To prove that it is maximal, we have to show that equality cannot
occur among any two subsequent elements. To see this, observe that at each step
$k$:
\begin{equation}\label{eq:proof}
\iota_k\not\leq m_k,\quad \iota_k\leq m_{k-1},\iota_k\leq
m_{k-2},\ldots,\iota_k\leq m_1.
\end{equation}
Hence $\iota_{k-1}\not\leq \iota_k$, otherwise $\iota_{k-1}\leq m_{k-1}$ would
hold, a contradiction. Hence, the sequence $\iota_n,\iota_{n-1},\ldots,\iota_1$
is non decreasing, and equality cannot occur. 

Let us prove that it suffices to choose $\iota_k$ as the smallest in
$\eta(n(j_{k-1})\setminus \eta(n(j_k))$. If at step $k$, a smallest $\iota_k$
is not chosen in $\eta(n(j_{k-1})\setminus \eta(n(j_k))$, it will be taken
after, and the sequence $\iota_n,\iota_{n-1},\ldots,\iota_1$ will be no more
non decreasing, a contradiction.

It remains to prove that $\pi$ is strictly increasing and to verify the
expression of $m$. Let us prove by induction that 
\begin{equation}\label{eq:pi1}
\pi(j_k)=1-m(\iota_n)-m(\iota_n\vee\iota_{n-1}) - \cdots-
m(\iota_n\vee\cdots\vee \iota_{k+1}),\quad k=n,\ldots,1.
\end{equation}
We show it for $k=n$. We have
\[
\pi(j_n)=1-\nu(m_n) = 1-\sum_{\substack{x\leq m_n\\x\in C_\pi}}m(x).
\]
Since $\iota_n\not\in\eta(m_n)$, no $x$ in $C_\pi$ can be smaller than $m_n$.
Hence $\pi(j_n)=1$. Let us assume (\ref{eq:pi1}) is true from $n$ up to some
$k$, and prove it is still true for $k-1$. Using (\ref{eq:proof}), we have:
\begin{align*}
\pi(j_{k-1}) &=1-\nu(m_{k-1}) = 1-\sum_{\substack{x\leq m_{k-1}\\x\in
C_\pi}}m(x) \\ 
& = 1 - \sum_{\substack{x\leq m_{k}\\x\in C_\pi}}m(x) -
m(\iota_n\vee\cdots\vee\iota_k) = \pi(j_k)-m(\iota_n\vee\cdots\vee\iota_k),
\end{align*}
which proves (\ref{eq:pi1}). Lastly, remark that the linear system of $n$
equations (\ref{eq:pi1}) is triangular, with no zero on the diagonal. Hence it
has a unique solution, which is easily seen to be (\ref{eq:mobpi}).
\end{proof}

As illustration of the theorem, we give an example.
\begin{example}
Let us consider the distributive autodual lattice given on
Fig. \ref{fig:lat3}. 
\begin{figure}[htb]
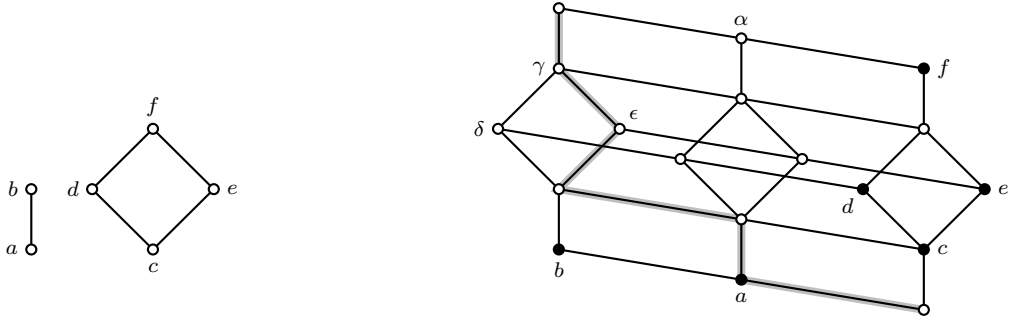

\begin{center}
\psset{unit=0.8cm}
\pspicture(0,0)(5,5)
\psline(0,1)(0,2)
\pspolygon(2,1)(1,2)(2,3)(3,2)
\pscircle[fillstyle=solid](0,1){0.1}
\pscircle[fillstyle=solid](0,2){0.1}
\pscircle[fillstyle=solid](2,1){0.1}
\pscircle[fillstyle=solid](1,2){0.1}
\pscircle[fillstyle=solid](2,3){0.1}
\pscircle[fillstyle=solid](3,2){0.1}
\uput[180](0,1){\scriptsize  $a$}
\uput[180](0,2){\scriptsize  $b$}
\uput[-90](2,1){\scriptsize  $c$}
\uput[180](1,2){\scriptsize  $d$}
\uput[0](3,2){\scriptsize  $e$}
\uput[90](2,3){\scriptsize  $f$}
\endpspicture
\hspace*{2cm}
\psset{unit=0.8cm}
\pspicture(0,0)(9,6)
%maximal chain
\psline[linewidth=3pt,linecolor=lightgray](7,0)(4,0.5)(4,1.5)(1,2)(2,3)(1,4)(1,5)
\pspolygon(1,2)(0,3)(1,4)(2,3)
\pspolygon(4,1.5)(3,2.5)(4,3.5)(5,2.5)
\pspolygon(7,1)(6,2)(7,3)(8,2)
\psline(1,1)(1,2)
\psline(4,0.5)(4,1.5)
\psline(7,0)(7,1)
\psline(1,4)(1,5)
\psline(4,3.5)(4,4.5)
\psline(7,3)(7,4)
\psline(1,1)(7,0)
\psline(1,2)(7,1)
\psline(0,3)(6,2)
\psline(2,3)(8,2)
\psline(1,4)(7,3)
\psline(1,5)(7,4)
\pscircle[fillstyle=solid,fillcolor=black](1,1){0.1}
\pscircle[fillstyle=solid,fillcolor=black](4,0.5){0.1}
\pscircle[fillstyle=solid](7,0){0.1}
\pscircle[fillstyle=solid](1,2){0.1}
\pscircle[fillstyle=solid](4,1.5){0.1}
\pscircle[fillstyle=solid,fillcolor=black](7,1){0.1}
\pscircle[fillstyle=solid](0,3){0.1}
\pscircle[fillstyle=solid](3,2.5){0.1}
\pscircle[fillstyle=solid,fillcolor=black](6,2){0.1}
\pscircle[fillstyle=solid](2,3){0.1}
\pscircle[fillstyle=solid](5,2.5){0.1}
\pscircle[fillstyle=solid,fillcolor=black](8,2){0.1}
\pscircle[fillstyle=solid](1,4){0.1}
\pscircle[fillstyle=solid](4,3.5){0.1}
\pscircle[fillstyle=solid](7,3){0.1}
\pscircle[fillstyle=solid](1,5){0.1}
\pscircle[fillstyle=solid](4,4.5){0.1}
\pscircle[fillstyle=solid,fillcolor=black](7,4){0.1}
\uput[-90](4,0.5){\scriptsize  $a$}
\uput[-90](1,1){\scriptsize  $b$}
\uput[0](7,1){\scriptsize  $c$}
\uput[225](6,2){\scriptsize  $d$}
\uput[0](8,2){\scriptsize  $e$}
\uput[0](7,4){\scriptsize  $f$}
\uput[90](4,4.5){\scriptsize  $\alpha$}
\uput[180](1,4){\scriptsize  $\gamma$}
\uput[180](0,3){\scriptsize  $\delta$}
\uput[45](2,3){\scriptsize  $\epsilon$}
\endpspicture
\end{center}
\caption{Example of autodual distributive lattice $L$ (right), with
  $\mathcal{J}(L)$ (left)}
\label{fig:lat3}
\end{figure}
Join-irreducible elements are $a,b,c,d,e,f$, while
meet-irreducible ones are $\alpha,b,\gamma,\delta,\epsilon,f$. We propose as
$\vee$-negation the following:
\begin{center}
\begin{tabular}{|c|c||c|c|}\hline
$x$ & $n(x)$ & $x$ & $n(x)$ \\ \hline
$a$ & $\alpha$ & $d$ & $\epsilon$\\
$b$ & $f$ & $e$ & $\delta$\\
$c$ & $\gamma$ & $f$ & $b$\\ \hline
\end{tabular}
\end{center}
Let us consider a possibility distribution satisfying
\[
\pi(c)<\pi(d)<\pi(e)<\pi(a)<\pi(f)<\pi(b)=1.
\]
(observe that the sequence $c,d,e,a,f,b$ is non decreasing, as requested). We
apply the procedure of Th. \ref{th:2}. For $b$, we have $n(b)=f=c\vee d\vee
e\vee f$, and for $f$, we have $n(f)=b=a\vee b$. Hence the first
join-irreducible element of the sequence, $\iota_6$, is $a$ (not in $\eta(f)$,
and minimal in $\eta(b)$). Table \ref{tab:1} summarizes all the steps.
\begin{table}[htb]
\begin{center}
\begin{tabular}{|c|c|c|c|c|c|}\hline
step $k$ & $x$ & $n(x)$ & $\eta(n(x))$ & $\iota_k$ & chain\\ \hline
6 & $b$ & $f$ & $c,d,e,f$ & $a$  & $a$\\
5 & $f$ & $b$ & $a,b$ & $c$  & $a\vee c$\\
4 & $a$ & $\alpha$ & $a,c,d,e,f$ & $b$  & $a\vee c\vee b$\\
3 & $e$ & $\delta$ & $a,b,c,d$ & $e$  & $a\vee c\vee b\vee e$\\
2 & $d$ & $\epsilon$ & $a,b,c,e$ & $d$  & $a\vee c\vee b\vee e\vee d$\\
1 & $c$ & $\gamma$ & $a,b,c,d,e$ & $f$  & $\top$\\ \hline
\end{tabular}
\end{center}
\caption{Computation of $C_\pi$}
\label{tab:1}
\end{table}
The maximal chain is in gray on Fig. \ref{fig:lat3}. 
We deduce that:
\begin{align*}
  \pi(b) &=1\\
  \pi(f) &=1-m(a)\\
  \pi(a) &=1-m(a) -m(a\vee c)\\
  \pi(e) &=1-m(a) -m(a\vee c) - m(a\vee c\vee b)\\
  \pi(d) &=1-m(a) -m(a\vee c) - m(a\vee c\vee b) - m(a\vee c\vee b\vee e)\\
  \pi(c) &=1-m(a) -m(a\vee c) - m(a\vee c\vee b) - m(a\vee c\vee b\vee e) -
  m(a\vee c\vee b\vee e\vee d)
\end{align*}
from which we deduce
\begin{align*}
m(a) &= \pi(b)-\pi(f)\\
m(a\vee c ) &= \pi(f) -\pi(a)\\
m(a\vee c\vee b) &= \pi(a)-\pi(e)\\
m(a\vee c\vee b\vee e) &= \pi(e) -\pi(d)\\
m(a\vee c\vee b\vee e\vee d) &= \pi(d) -\pi(c)\\ 
\end{align*}
and $m(\top) = 1-m(a) -m(a\vee c) - m(a\vee c\vee b) - m(a\vee c\vee b\vee e) -
  m(a\vee c\vee b\vee e\vee d)= \pi(c)$. 
\end{example}

\section{Acknowledgment}
The author addresses all his thanks to Bruno Leclerc and Bernard Monjardet for
fruitful discussions on negations.  

\bibliographystyle{plain}
\bibliography{../BIB/fuzzy,../BIB/grabisch,../BIB/general}

\end{document}